\documentclass[runningheads]{llncs}
\usepackage[utf8]{inputenc}
\usepackage[T1]{fontenc}
\usepackage {stmaryrd,amsfonts,amsmath,amssymb,xcolor,bm,dsfont,mathtools,float,tikz,tabularx,thm-restate,ntheorem,cite}
\usepackage[inline]{enumitem}
\usepackage[all]{xy}
\usepackage{xspace}
\usepackage[hidelinks = true,colorlinks = true,citecolor = blue, urlcolor = blue,linkcolor = blue]{hyperref}
\usepackage{cleveref}
\usetikzlibrary[arrows,decorations.pathmorphing]
\usepackage[misc,geometry]{ifsym}

\input symbexfree.sty \IfFileExists{orcid.pdf}{
 \def\orcidsymbol{\includegraphics[height=9pt]{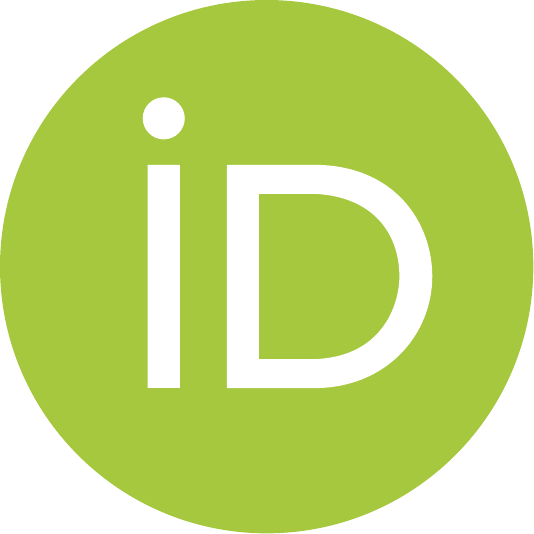}}
}{
 \def\orcidsymbol{\textcolor{lipicsGray}{\fontsize{9}{12}\sffamily\bfseries \faOrcid}}
}
\def\orcidID#1{\hspace{2pt}\smash{\href{http://orcid.org/#1}{\protect\raisebox{0pt}{\protect\orcidsymbol}}}}

\title{Correct and Complete\\ Symbolic Execution for Free\thanks{This research is partially supported by the NWO grant No.~OCENW.M20.053, ERC grant Autoprobe (no. 101002697) and a Royal Society Wolfson fellowship.}}
\author{{Erik Voogd~\Letter\inst{1}\orcidID{0009-0007-9712-3224}}
  \and {Einar Broch Johnsen\inst{1}\orcidID{0000-0001-5382-3949}}
  \and {\\Åsmund Aqissiaq Arild Kløvstad\inst{1}\orcidID{0009-0007-1957-4409}}
  \and {Jurriaan Rot\inst{2}\orcidID{0000-0002-1404-6232}}
  \and {Alexandra Silva\inst{3}\orcidID{0000-0001-5014-9784}}}
\institute{University of Oslo, Oslo, Norway 
\and Radboud University, Nijmegen, the Netherlands 
\and Cornell University, Ithaca NY, USA\\
Corresponding author: \Letter~\email{erikvoogd@live.nl}}
\authorrunning{E.~Voogd et al.}

\begin{document}
\maketitle \begin{abstract}
 Symbolic execution is a powerful technique for program analysis.
 However, the formal semantics underlying symbolic execution is often developed on an ad-hoc basis and decoupled from the concrete semantics of the programming language.
 To overcome this issue, we introduce \emph{symbolic SOS}: a rule format that allows us to simultaneously specify concrete \emph{and} symbolic operational semantics.
 We prove that symbolic semantics, when generated from symbolic SOS, is both correct and complete with respect to the corresponding concrete semantics.
 The approach relies only on an algebraic signature of the source language, and is thus language-independent.
\end{abstract}

\section{Introduction} \label{sec:intro} 
Symbolic execution is an established program analysis technique with a long history~\cite{king1976symbolic,boyer1975select,katz1975towards}, that is widely used for bug finding, verification and even program synthesis~\cite{baldoniSurveySymbolicExecution2018,fragososantosGillianPartMultilanguage2020,fragososantosGillianPartMultilanguage2021,ahrendt16key,porncharoenwase2022withmerging}.
More recently, the systematic study of the formal foundations of symbolic execution in its own right has gained increasing interest~\cite{arusoaie2013generic,lucanu2017generic,deboerSymbolicExecutionFormally2021,steinhofel2020abstract}.
In particular, De Boer and Bonsangue~\cite{deboerOnTheNature2019,deboerSymbolicExecutionFormally2021}
introduce the notions of \emph{correctness} and \emph{completeness}, formalizing a correspondence between symbolic execution and concrete program execution.

Proving that a given symbolic semantics is correct and complete, however, is a tedious task that would benefit from automation.
This paper provides the basis for such automation by answering the following question: \emph{what are sufficient conditions for a language specification to define a symbolic semantics that is correct and complete by design?}
We take language specifications to be defined in Structural Operational Semantics (SOS)~\cite{Plotkin04a,plotkin1997mathematical}, a standard formalism to specify programming language semantics as collections of inference rules.
Our goal is to extract, for a given language with an operational semantics expressed as an SOS specification, a correct and complete symbolic semantics.

Our starting point is the \emph{GSOS} rule format~\cite{bloom1995bisim}.
The syntactic restrictions on rules in GSOS specifications ensure that the resulting semantics is well-defined and compositional with respect to bisimilarity.
The GSOS
format was later generalized by Turi and Plotkin into \emph{abstract GSOS}~\cite{plotkin1997mathematical}, which formalizes specifications through interaction between algebra (representing syntax) and coalgebra (representing transition systems).
Goncharov et al. have recently introduced \emph{stateful structural operational semantics} (SSOS)~\cite{goncharov2022stateful}, which adapts abstract GSOS to a stateful setting.

In order to automatically obtain symbolic semantics from SOS specifications, we refine SSOS to \emph{symbolic SOS}.
A symbolic SOS specification defines both a concrete and a symbolic operational semantics, in terms of two (small-step) transition systems.
Our main result is that this symbolic semantics is always correct and complete with respect to the concrete semantics.
Any programming language defined in the symbolic SOS format thereby comes equipped with a corresponding correct and complete symbolic semantics.
To prove this general correspondence result, we introduce the notion of \emph{syncrete bisimulation}, which gives a sufficient condition for correctness and completeness.
We then show that for every symbolic SOS specification, the induced concrete and symbolic semantics are bisimilar, and therefore the symbolic semantics is correct and complete.

\paragraph{Contributions}
In summary, this paper introduces:
\begin{enumerate*}[label=(\roman*)]
 \item the \emph{symbolic SOS} rule format (\Cref{sec:ruleformat}) from which both a concrete and a symbolic semantics can be derived;
 \item the notion of \emph{syncrete bisimulation}, which provides a coinductive proof technique for correctness and completeness (\Cref{sec:symbolicsemantics}); and \item a justification for the rule format (\Cref{thm:syncretebisim})
 stating that the two derived semantics from a symbolic SOS
 specification are always syncretely bisimilar.
\end{enumerate*}

\section{Overview}\label{sec:overview}
This section summarizes the technical development and presents our key results.
We also introduce our running example: the imperative toy language \While.
\While serves as a minimal concrete example to illustrate our results, but the approach is language-independent.
In \Cref{sec:extensions}, we consider additional common programming constructs, to showcase the power of our approach.

\begin{example}[The Syntax of\, \While]\label{ex:grammar}
 The syntax of \While is given by three grammars for expressions, Boolean expressions and program statements.
 Let us consider the following grammar for expressions and Boolean expressions:
 $$\exp \, ::= \, \x \,\mid\, n \,\mid\, \exp + \exp \,\mid\, \exp - \exp \,\mid\, \exp * \exp \qquad \bexp \, ::= \, \true \,\mid\, \neg\bexp \,\mid\, \bexp\land\bexp \,\mid\, \exp<\exp \,\mid\, \exp=\exp $$
 where $\x\in\Var$ ranges over program variables and $n\in\Z$ over integers.
 The syntax of the \While programming language is given by the grammar $$ \p \, ::= \, \pskip \, \mid \, \assign \, \mid \, \seq\p\p \, \mid \, \pif\bexp\p\p \, \mid \, \while\bexp\p$$
 These program statements represent inaction, assignments, sequencing, branching, and unbounded iteration.
\end{example}

The semantics of \While can be specified using an SOS format to define a transition system that formalizes the evolution of program configurations.
In particular, the SSOS
format~\cite{goncharov2022stateful} considers pairs $(p,\state)$, where $p$ is the program to be executed and $\state$ is a state.
Usually, states associate values to program variables.
Consider, for example, the following rules for sequencing and assignment:
$$
\srule{}{(\p,\state)\contrterm\state'}{(\seq\p\q,\state)\contr(\q,\state')}
\qquad \begin{array}{c}\quad\\
 \saxiom{}{(\assign,\state)\contrterm \state[\x \mapsto \state(\exp)]}
\end{array}
$$
The transition relation $\contr$ denotes one step of \emph{concrete execution}, evolving to a new program and state, and $\contrterm$ denotes termination.
The subscript~$c$ here emphasizes that system states are \emph{concrete}.
With $[\x\mapsto\state(\exp)]$, an assignment updates the value of variable $\x$ to $\state(\exp)$, which is the expression $\exp$ evaluated in the state $\state$ according to a standard interpretation of arithmetic operations.
Each rule actually represents a (potentially infinite) family of rules, one for each state.
Together, these two (families of) rules describe a concrete execution model consisting of sequences of assignments to program variables.

For symbolic execution we may also define a \emph{symbolic} semantics, again using the SSOS format.
This is defined as a transition relation between states~$\sstate$ which are substitutions, i.e., they associate expressions to variables.
The symbolic state moreover contains a set of path constraints which we add later.
Symbolic rules are often identical to their concrete counterpart, up to the interpretation of states.
Indeed, consider the analogous symbolic rules for sequencing and assignments:
$$\srule{}{(\p,\sstate)\symbextrterm\sstate'}{(\seq\p\q,\sstate)\symbextr(\q,\sstate')}
\qquad \saxiom{}{(\assign,\sstate)\symbextrterm \sstate[\x \mapsto \app\sstate\exp]}
$$
Here, the \emph{symbolic} transition relation $\symbextr$ denotes one step of \emph{symbolic execution}, evolving to a new program and symbolic state, and $\symbextrterm$ denotes termination.

The rules look identical to their concrete counterparts, but the state updates differ subtly.
While the concrete rule updates the state to map the variable~$\x$ to a new value, namely the expression~$\exp$ evaluated in state~$\state$, the symbolic rule updates the state to map the variable~$\x$ to a new expression:
$\app\sstate\exp$ is the expression~$\exp$ with all variables substituted according to~$\sstate$.

For symbolic execution to be useful, it must indeed be an abstraction of the concrete execution.
That is, informally, for each concrete step there must be a symbolic step whose final state describes the state update of the concrete step.
For some rules in concrete execution, however, it is unclear what the matching symbolic step should look like.
Consider the concrete rules for branching:
$$\srule{}{\state\sat\bexp}{(\pif\bexp\p\q,\state) \contr (\p,\state)} \quad \srule{}{\state\sat\neg\bexp}{(\pif\bexp\p\q,\state) \contr (\q,\state)}
$$
These rules express that the program $\pif\bexp\p\q$ evolves to $\p$
if $\state$ satisfies $\bexp$, and to $\q$ otherwise; the state $\state$ remains unaltered in either case.
These rules are deterministic: branching is resolved by checking whether or not states satisfy the Boolean expression~$\bexp$, guarding the control-flow of the program.

From a \emph{symbolic} state, however, the question of whether the symbolic state satisfies $\bexp$ cannot be resolved.
We may enable \emph{both} transitions 
$$\srule{}{}{(\pif\bexp\p\q,\sstate) \symbextr (\p,\sstate)} \qquad \srule{}{}{(\pif\bexp\p\q,\sstate) \symbextr (\q,\sstate)}$$
rendering the symbolic transition relation non-deterministic.
Starting from a single state, therefore, a program generates many different sequences of transition steps, called \emph{paths}.
In contrast, concrete execution generates a \emph{single} path.
So which symbolic execution path simulates the concrete execution path?

To answer this question, states in symbolic execution are enhanced with a \emph{path condition} and the transition relation now relates triples of programs, states, and a path condition.
Path conditions aggregate the Boolean conditions that guide a program's control flow under the current substitution.
In the case of conditional branching, this is realized through the following two rules:
$$\srule{}{}{(\pif\bexp\p\q,\sstate,\pc) \symbextr (\p,\sstate,\pc\wedge\app\sstate\bexp)} \quad \srule{}{}{(\pif\bexp\p\q,\sstate,\pc) \symbextr (\q,\sstate,\pc\wedge\app\sstate{\neg\bexp)}}
$$
Both steps augment the path condition by substituting for the variables $\x$ their associated expressions $\app\sstate\x$ in the expressions~$\bexp$ and~$\neg\bexp$.\footnote{We assume that this always results in a Boolean expression; abstracting from program-level type correctness.} 
The resulting system is still technically non-deterministic, but the path condition ``determinizes'' the symbolic execution by specifying exactly which concrete executions it simulates.

To argue that the path condition of a symbolic execution path is indeed a precondition for concrete executions, the two types of executions are connected by proving \emph{correctness} and \emph{completeness}.
We define these notions following De Boer and Bonsangue~\cite{deboerSymbolicExecutionFormally2021}.
Below, the initial configuration ${(\sstate_0,\pc_0)=(\id,\true)}$ consists of the identity substitution $\sstate_0$ on variables and the Boolean \emph{truth} formula~$\true$.

\subsubsection{Correctness.}
Symbolic execution is \emph{correct} with respect to concrete execution if all \emph{symbolic} execution paths 
\begin{equation}\label{eq:symbexpath}
	{(\p_0,\sstate_0,\pc_0) \symbextr (\p_1,\sstate_1,\pc_1) \symbextr \dots \symbextr (\p_k,\sstate_k,\pc_k)}
\end{equation}
simulate the \emph{concrete} execution paths from $(\p_0,\state_0)$ for which $\state_0\sat\pc_k$.
Formally,
\begin{equation}(\p_0,\state_0\after\sstate_0) \contr (\p_1,\state_0\after\sstate_1) \contr \dots \contr (\p_k,\state_0\after\sstate_k) 
	\end{equation}
is the \emph{unique} concrete execution path starting from $\p_0$ and $\state_0\sat\pc_k$.
Here, $\state\after\sstate$ denotes \emph{evaluation of $\state$ after substitution $\sstate$}---this is made formal in \Cref{def:substeval}.

\subsubsection{Completeness.}
Symbolic execution is \emph{complete} with respect to concrete execution if every concrete execution path 
\begin{equation} (\p_0,\state_0) \contr (\p_1,\state_1) \contr \dots \contr (\p_k,\state_k) 
\end{equation}
is simulated by a symbolic execution path as in Equation~\eqref{eq:symbexpath},
whose resulting path conditions are satisfied by~$\state_0$, i.e., $\state_0\sat\pc_j$ for all $j\in[0..k]$.

Correct- and completeness are properties of symbolic execution that are defined with respect to the concrete semantics:
\emph{correctness} is when every symbolic execution path corresponds to a realizable concrete computation, and \emph{completeness} is when all concrete paths are represented by some symbolic path.
In the context of program analysis, on the other hand, correctness is about \emph{coverage} and completeness is about \emph{precision}~\cite{arusoaie2013generic,deboerSymbolicExecutionFormally2021}.

The process of defining concrete and symbolic semantics separately and then proving correctness and completeness is cumbersome.
As observed, the symbolic and concrete rules are (almost) identical, leading to our key question:

\medskip \noindent 
\textbf{Question:} \emph{Can we obtain correct and complete symbolic semantics directly from a language specification?}

\medskip \noindent We answer this question by defining a \emph{symbolic} SOS rule format.
A key insight for our work is the observation that both concrete and symbolic transition systems can arise from the same underlying specification, if this specification is sufficiently structured to obtain both semantics.
In particular, the resulting state needs to be carefully constructed to ensure symbolic simulation, and care must be taken to allow building a path condition.

In symbolic SOS specifications (\Cref{sec:ruleformat}), states are meta-variables, i.e., placeholders for both concrete \emph{and} symbolic states.
In \Cref{sec:concretesemantics} we make explicit how the concrete semantics is derived from a symbolic SOS specification.
Then, in \Cref{sec:symbolicsemantics}, we derive the symbolic semantics and we define the notion of \emph{syncrete bisimulations}.
Crucially, we show that the derived semantics are related by a syncrete bisimulation relation, ensuring both correctness and completeness.

\section{Preliminaries}\label{sec:preliminaries}
\vspace*{-0.2cm}
For our programming language, we take a signature $\signature$ and a fixed set of program variables $\Var$.
The signature consists of ways to generate \rom 1 \emph{programs} using operators in $\Sigma$; \rom 2 \emph{expressions} using $\ExpOp$-operators; \rom 3 \emph{predicates} using $\PredOp$; and \rom 4 \emph{Boolean expressions} using $\LogOp$. 
Every operator has some arity given by $\arity \colon \Sigma+\ExpOp+\PredOp+\LogOp \to \N$.
The meaning of $\guardsign$ is explained below.

For any set $X$, write $\EE X$ for the set of terms over $X$ using operations in $\ExpOp$; we call these \emph{expressions} over $X$. 
Then $\EE\Var$ is the set of \emph{program} expressions, used in, e.g., assignments.
With $\Val$ a set of values (integers, rationals, lists, etc.), a family $\evaluate$ of maps $(\evaluate_\operator \colon \Val^{\arity(\operator)} \to \Val)_{\operator\in\ExpOp}$ interprets $\ExpOp$-operators.
The inductive extension $\lift\evaluate$ of $\evaluate$ over itself, i.e., $\lift\evaluate\colon\EE\Val \to \Val$ with $\lift\evaluate(\val)=\evaluate(\val)$, $\val\in\Val$ and $\lift\evaluate(\operator(\exp_1,\dots,\exp_{\arity(\operator)}) = \evaluate_{\operator}(\lift\evaluate(\exp_1),\dots,\lift\evaluate(\exp_{\arity(\operator)}))$ evaluates expressions over values.
We will write $\evaluate$ for $\lift\evaluate$.

\emph{States} during concrete execution (or \emph{concrete states}) are mappings $\state \colon \Var \to \Val$ that assign values to program variables.
Given $\evaluate\colon\EE\Val\to\Val$, program expressions $\exp \in \EE\Var$ can be evaluated in any state by inductive extension of the map $\state\colon\Var\to\Val$ to $\lift\state\colon\EE\Var \to \Val$ defined by $\lift\state(\x)=\state(\x)$ and $\lift\state(\operator(\exp_1,\dots,\exp_{\arity(\operator)}))=\evaluate_\operator(\lift\state(\exp_1),\dots,\lift\state(\exp_{\arity(\operator)}))$.
Then $\lift\state=\lift\evaluate\circ\EE\state$, where $\EE\state\colon\EE\Var\to\EE\Val$ performs uniform substitution of variables by values in the expression according to $\state$.
Sometimes we write $\state$ for $\lift\state$.

The set $\PP X$ of \emph{predicates} over $X$ consists of expressions $\predicate(\exp_1,\dots,\exp_n)$, where $\predicate \in \PredOp$ is an $n$-ary predicate operator, i.e., $\arity(\predicate)=n$, and each $\exp_i \in \EE X$ is an expression over~$X$.
Common examples include membership, equality, and inequalities.
Assume we can \emph{interpret} predicates over values using ${\interpret \colon \PP\Val \to \{\truth,\falsehood\}}$.
We say a concrete state $\state\in\Val^\Var$ \emph{satisfies} a program predicate ${\predicate(\exp_1,\dots,\exp_n)\in\PP\Var}$, written $\state \sat_{\evaluate,\interpret} \predicate(\exp_1,\dots,\exp_n)$ iff $\interpret(\predicate(\lift\state(\exp_1),\dots,\lift\state(\exp_n)))=\truth$.
We will leave dependence on $\evaluate$ and $\interpret$ implicit by just writing $\state \sat \predicate(\exp_1,\dots,\exp_n)$.

With $\LogOp$ as a set of logical operators such that every $n$-ary $\logop\in\LogOp$ has a truth table $\{\truth,\falsehood\}^n \to \{\truth,\falsehood\}$, $(\ExpOp,\PredOp,\LogOp)$ is a first-order logic signature without quantification.
Let $\BB X$ be the set of \emph{Boolean expressions} over $X$ inductively generated by operators in $\LogOp$ over predicates in $\PP X$.
We assume $\LogOp$ contains binary conjunction $\land$ and disjunction $\lor$, and unary negation $\neg$, each with its conventional truth table. 
The constant $\true\in\LogOp$ (true) is satisfied by all states, and $\false\in\LogOp$ (false) by no states.
Write $\state\sat\bexp$ if a state $\state\in\Val^\Var$ satisfies $\bexp\in\BB\Var$.

The set $\Sigma^*(X)$ of \emph{programs} with free variables $X$ is
inductively generated by operators in~$\Sigma$.  Let
$\Prg=\Sigma^*(\emptyset)$ be the set of \emph{(closed)} programs.
These are the only programs that will be equipped with a semantics.
Operators in $\Sigma$ may use \emph{expressions} in $\EE\Var$ and
\emph{Boolean expressions} in $\BB\Var$.  Crucial to our work,
$\guardsign \colon \Sigma \to \BB\Var$ assigns to every operator
$\operator\in\Sigma$ an associated \emph{guard}
$\guardof\operator \in \BB\Var$ governing control flow of the
semantics.  For the rest of the paper, we consider the signature
$\signature$, the set of variables $\Var$, and the interpretations
$\evaluate$ and $\interpret$ fixed.

\begin{example}[The \While Signature]\label{ex:whilesignature}
 The signature $\signature$ for the \While language from \Cref{ex:grammar} is as follows:
 $\Sigma$ contains the constant $\pskip \in \Sigma$, a binary operator for sequencing, and one constant for every pair $(\x,\exp) \in \Var \times \EE\Var$ of variable and expression representing an assignment $\assign$.
 Each of these operators has $\true$ as guard.
 $\Sigma$ furthermore contains, for each Boolean expression ${\bexp\in\BB\Var}$, a binary operator $\pif{\bexp}{\,\cdot\,}{\,\cdot\,}$ for branching with guard $\guardof{\pif{\bexp}{\,\cdot\,}{\,\cdot\,}}=\bexp$ and a unary operator $\while\bexp{\,\cdot\,}$ for unbounded iteration with guard $\guardof{\while\bexp{\,\cdot\,}}=\bexp$.
 The set $\ExpOp$ contains the binary operators $+$, $-$, $*$, and all integers.
 The set $\PredOp$ contains binary $<$ and $=$, and $\LogOp$ contains constant $\true$, unary $\neg$, and binary $\land$.
 We let $\Val=\Z$ and $\evaluate$ and $\interpret$ are standard; e.g., $\evaluate(6 * 7) = 42$ and $\interpret(6*7>0\lor\false) = \truth$.
\end{example}

\section{Symbolic Rule Format} \label{sec:ruleformat}
In this section, we introduce the \emph{symbolic SOS} rule format.
 The format is purely syntactic, to the point that meta-variables are used as placeholders for both programs and states.
 This allows us to derive both concrete semantics~(\Cref{sec:concretesemantics}) and symbolic semantics~(\Cref{sec:symbolicsemantics}) from a single specification in the format. 

Let $\XVar=\{\mx,\my,\dots\}$ and $\SVar=\{\ma,\mb,\dots\}$ be sets of meta-variables that are placeholders for programs and states, respectively.
An (uninterpreted) \emph{transition} is either progressive or terminating.
A \emph{progressive} transition is an expression of the form $(\mx,\ma) \metatr (\my,\mb)$ with $\mx,\my \in \XVar$ and $\ma,\mb\in \SVar$.
Here, $\mx$ is called the \emph{source} and is said to transition to the \emph{target} $\my$ with \emph{input} $\ma$
and \emph{output} $\mb$.
A \emph{terminating} transition lacks a target and only produces output, written $(\mx,\ma) \metatrterm \mb$.
\emph{Terms} $\mt \in \Sigma^*(\XVar)$ may be used as sources and targets.
Using a special termination symbol $\terminated\in\Sigma$, we let $(\mx,\ma)\metatrterm \mb$ be synonymous to $(\mx,\ma) \metatr (\terminated,\mb)$.
We use $\literal$ to denote uninterpreted transitions, progressive or terminating.
An SOS \emph{rule} consists of a set $\{\literal_i\}_{i=1..n}$, called the \emph{premises}, together with a \emph{conclusion}~$\literal$.
\vspace*{-0.4cm}
\begin{definition}[Symbolic SOS Rule] \label{def:symsosrule} A \emph{symbolic SOS rule} for an operator $\operator \in \Sigma$ of arity $n=\arity(\operator)$ is a rule 
 $ \srule{}{\literal_1 \quad \dots \quad \literal_{n} \quad\quad \trigger}{\literal}$
 where $\trigger\in\BB\Var$ is a Boolean expression called the \emph{trigger}  of the rule and 
 \begin{itemize}
 \item the source of the conclusion $\literal$ is $\operator(\mx_1,\dots,\mx_n)$;
 \item the source of the premise $\literal_i$ (each $i$) is $\mx_i$ and its target (if progressive) is $\my_i$;
 \item the input of each premise $\literal_i$ and the conclusion input is $\ma$;
 \item the output of premise $\literal_i$ is $\mb_i$;
 \item if $\literal$ is progressive, its target is in $\Sigma^*(\{\mx_1,\dots,\mx_n\} \cup \{\my_i \mid \literal_i \textup{ is progressive}\})$;
 \item the conclusion output is a map $\Var \to \EE(\{\ma,\mb_1,\dots,\mb_n\}\times\Var\})$.
 \end{itemize}
\end{definition}
This definition is an adaptation of the format of \emph{stateful SOS} \cite{goncharov2022stateful}; we have replaced states by meta-variables, added extra structure in the conclusion output, and added a Boolean trigger for control-flow.
Our specification, defined just below, requires the trigger to be either the guard $\guardof\operator$ of $\operator$, or its negation.
The last item says that the conclusion output can store an expression over $\SVar\times\Var$ for every variable, but it restricts to using meta-variables in $\SVar$ occurring in the rule.
An expression in $\EE(\SVar\times\Var)$ can be interpreted as a value in $\Val$ once the meta-variables have been replaced by concrete states---this is technically outlined in \Cref{sec:concretesemantics}---turning the conclusion output into a new concrete state.
Replacing the meta-variables by \emph{symbolic} states (defined later), the map $\Var\to\EE(\SVar\times\Var)$ can be interpreted as a new symbolic state, as outlined in \Cref{sec:symbolicsemantics}.

A \emph{symbolic SOS specification} requires each operator to have exactly \emph{two} rules whose triggers are complementary.
Letting these triggers coincide with an operator's guard and its complement, one rule has trigger $\guardof\operator$, the other $\neg\guardof\operator$.
The behavior of a program may also depend on whether or not any of its subterms terminates (e.g. sequencing).
There are therefore two rules for every operator $\operator$ \emph{and} every set $W \subseteq \{1,\dots,\arity(\operator)\}$ indicating which premises are progressive.
\begin{definition}[Symbolic SOS Specification] \label{def:symsos} 
 A \emph{symbolic SOS specification} for a signature $\signature$ is a set $\symbexspec$ of symbolic SOS rules with the following condition: for every operator $\operator \in \Sigma\setminus\{\terminated\}$ with $n=\arity(\operator)$ and for every subset $W \subseteq \{1,\dots,n\}$, there are exactly \emph{two} rules $\RR_1,\RR_2 \in \symbexspec$ such that
 \begin{itemize}
 \item the premises $\{\literal_i\}_{i\in W}$ of both $\RR_1$ and $\RR_2$ are progressive, and
 \item the premises $\{\literal_i\}_{i \in [1..n]\setminus W}$ of both $\RR_1$ and $\RR_2$ are terminating.
 \end{itemize}
 Moreover, for these rules, the triggers of $\RR_1$ and $\RR_2$ are $\guardof\operator$ and $\neg\guardof\operator$.
\end{definition}

\begin{figure}[!t]
\centering
$\begin{array}{l@{\qquad}l}
 \saxiom{skip}{(\pskip,\ma)\metatrterm \ma} &
 \saxiom{assign}{(\assign,\ma)\metatrterm \ma[\x \mapsto \ma(\exp)]} \\\\
 \srule{seq-0}{(\mx,\ma)\metatrterm \ma'}{(\seq\mx\my,\ma)\metatr(\my,\ma')} &
 \srule{seq-n}{(\mx,\ma)\metatr(\mx',\ma')}{(\seq\mx\my,\ma)\metatr(\seq{\mx'}\my,\ma')} \\\\
 \srule{if-T}{\bexp}{(\pif\bexp\mx\my,\ma) \metatr (\mx,\ma)} &
 \srule{while-T}{\bexp}{(\while\bexp\mx,\ma) \metatr (\seq\mx{\while\bexp\mx},\ma)} \\\\
 \srule{if-F}{\neg\bexp}{(\pif\bexp\mx\my,\ma) \metatr (\my,\ma)} &
 \srule{while-F}{\neg\bexp}{(\while\bexp\mx,\ma) \metatrterm \ma}
 \end{array}
$
 \caption{A symbolic SOS specification for the \While language,
   from which both symbolic and concrete semantics can be derived.}
 \label{fig:whilespec}
\end{figure}

\begin{example}[Specification for \While]\label{ex:specification}
 A symbolic SOS specification for the language system for \While (from \Cref{ex:whilesignature}) is shown in \Cref{fig:whilespec}.
 We omit the trigger of a rule if it is $\true$.
 The rule for $\neg\true$ does not matter for the semantics (see \Cref{sec:concretesemantics,sec:symbolicsemantics}), because no state ever satisfies~$\neg\true$.
 The rules for branching, iteration, and sequencing are syntactic sugar for sets of rules.
 \refRule{while-F}, for instance, represents two rules: one with premise $(\mx,\ma) \metatr (\mx',\ma')$ and one with premise $(\mx,\ma) \metatrterm \ma'$.
 Premise targets and outputs are not used in these instances.

 When writing $\ma$ (or $\ma'$) in the conclusion output of a rule, we mean the map $\Var \to \EE(\SVar \times\Var)$ that sends $\x$ to $(\ma,\x)$.
 \refRule{assign} uses common notation for function \emph{updates} $\ma[\x\mapsto\ma(\exp)]$, i.e., a function that maps every $\y$ to $(\ma,\y)$ except $\x$, which is mapped to $\ma(\exp)$.
 Here, $\ma(\exp)$ is shorthand for the expression $\exp$ with every variable $\y$ substituted by $(\ma,\y)$.

 In the following two sections, the meta-variables for states will be substituted by \emph{concrete} states (\Cref{sec:concretesemantics}) and by \emph{symbolic} states (\Cref{sec:symbolicsemantics}).
 There, the reasons for our choice of shorthand notations for $\ma$ and $\ma(\exp)$ will be made clear.
\end{example}

\section{Concrete Semantics} \label{sec:concretesemantics}
We now show how concrete semantics is derived from a symbolic SOS specification.
We make this derivation explicit to juxtapose it with the derivation of the symbolic semantics, and to show how the meta-variables can be formally interpreted as actual states, both concrete and symbolic.

Recall that $\Prg=\Sigma^*(\emptyset)$ is the set of all programs.
Usually, the meta-variables of language specification rules range over
all programs.
Formally, the symbolic SOS rules from the previous section are
equipped with a \emph{meta-substitution},\footnote{Not to be confused
  with substitutions of values or expressions.} i.e., a map
$\metasubst_X \colon \XVar \to \Prg$.  This mapping canonically
extends to $\Sigma^*(\XVar) \to \Prg$, performing uniform substitution
on programs over meta-variables.  Usually, a meta-substitution
$\metasubst_X$ is
partially defined, namely on the sources
$\{\mx_1,\dots,\mx_n\}$ and the targets
$\{\my_i\mid\literal_i \textup{ is progressive}\}$ of a rule's
premises $\literal_1,\dots,\literal_n$.  The rule format ensures that
$\metasubst_X$ can also be applied to the conclusion source
$\operator(\mx_1,\dots,\mx_n)$ and the conclusion target $\mt$.  The
meta-variables in the rules are all distinct, but since
meta-substitutions can be injective, programs in the rule may
coincide.

The key insight here is that states, much like programs, can also be interpreted symbolically using meta-variables:
\begin{definition}[Meta-Substitution of Concrete States] \label{def:metasubstconcrete} A \emph{meta-substitution of concrete states} is a map $\metasubst_S \colon \SVar \to \Val^\Var$.
\end{definition}
We will often combine meta-substitutions of programs $\metasubst_X \colon \XVar \to \Prg$ and of states $\metasubst_S \colon \SVar \to \Val^\Var$ into one meta-substitution $\metasubst \colon \XVar + \SVar \to \Prg + \Val^\Var$.

Meta-substitutions of states are usually only partially defined, namely on the input $\ma$---which is the same for all premises and for the conclusion---and on the premise outputs $\{\mb_1,\dots,\mb_n\}$ occurring in a rule.
The rule format guarantees that a meta-substitution of states $\metasubst_S$ can be applied to the rule's conclusion output $\mc \colon \Var \to \EE(\SVar\times\Var)$, which would be of type $\metasubst_S(\mc) \colon \Var \to \EE(\Val^\Var\times\Var)$.
But now that we have access to concrete states, we can use function evaluation ${\eval \colon \Val^\Var\times\Var \to \Val, (\state,\x)\mapsto\state(\x)}$, to interpret these pairs occurring in an expression as values from the concrete state.
Evaluating the expression of values for each variable using $\evaluate \colon \EE\Val\to\Val$ provides us with a new state given by 
$ \metasubst_S(\mc) \colon \Var \to \EE(\Val^\Var\times\Var) \xrightarrow{\EE(\eval)} \EE(\Val) \xrightarrow{\evaluate} \Val $
for the conclusion output.

A \emph{concrete execution model} is a deterministic unlabeled transition system $(\Prg\times\Val^\Var,\contr)$.
The concrete execution model $\contr$ \emph{intended} by a symbolic SOS specification $\symbexspec$ for the signature $\signature$ is recursively defined on the structure of programs in~$\Prg$ as follows:
\begin{definition}[Concrete Semantics] \label{def:concretesemantics} 
  Let $\operator \in \Sigma$ be an operator with ${n=\arity\operator}$, let $\p_1,\dots,\p_n \in \Prg$ be programs for which transitions have been defined, and let ${\state\in\Val^\Var}$ be an arbitrary state.
 Let ${W \subseteq \{1,\dots,n\}}$ indicate which transitions $(\p_i,\state)\contr(\p^{(i)},\state^{(i)})$ are progressive, and let ${\RR_1,\RR_2 \in \symbexspec}$ be the two rules for $\operator$ and $W$ with triggers $\guardof\operator$ and $\neg\guardof\operator$, respectively.
  Then let $(\operator(\p_1,\dots,\p_n),\state) \contr (\metasubst(\mt),\metasubst(\mc))$ by definition, with the meta-substitution 
  $$ \metasubst \colon \XVar+\SVar \to \Prg + \Val^\Var,\,\,\, 
  \mx \mapsto 
  \begin{cases} 
    \p_i & \mx=\mx_i\\
    \p^{(i)} & \mx=\mx'_i
  \end{cases} \quad 
  \mb \mapsto 
  \begin{cases} 
    \state & \mb=\ma\\
    \state^{(i)} & \mb=\mb_i 
  \end{cases}
  $$ 
  using conclusion target and output $(\mt,\mc)$ of $\RR_1$ if $\state\sat\guardof\operator$ and of $\RR_2$ otherwise.
\end{definition}
Constants in $\Sigma$ use axioms in the specification and constitute the base cases in this recursive definition.
The resulting relation is clearly deterministic: every pair $(\p,\state)$ defines exactly one outgoing transition, except when $\p=\terminated$.

\begin{example}\label{ex:concretestep}
 Consider a program in \While that computes absolute values:
 $$\pabs \triangleq \pif{(\x<0)}{\{\x\asgn 0-\x\}}{\{\pskip\}}$$
 We have a specification from \Cref{ex:specification} which induces a concrete execution model $(\contr,\Prg\times\Val^\Var)$.
 Let $\state$ be a concrete state that maps $\x$ to $-42$.
 Then $(\pabs,\state)\contr(\x\asgn 0-\x,\state) \contrterm \state',$
 where $\state' \colon \x \mapsto 42$.
 The number $42$ was obtained by evaluating the expression $0-(\ma,\x)$ after the meta-substitution $\metasubst_S \colon \ma \mapsto \state$.
\end{example}

\section{Symbolic Semantics} \label{sec:symbolicsemantics}
We develop the semantics of symbolic execution, based on the same specification as we used to derive the concrete semantics.
The meta-substitution will now substitute meta-variables by \emph{symbolic} states.
After describing symbolic states and revisiting meta-substitutions, we introduce \emph{syncrete bisimulation} to coinductively formalize correctness and completeness.
Symbolic execution semantics is derived from the same specification as the concrete semantics.
This semantics is both correct and complete with respect to the concrete semantics (\Cref{thm:syncretebisim}).

The domain of a symbolic state $\sstate\colon\Var\to\EE\Var$, like
that of a concrete state, can
be inductively extended from $\Var$ to $\EE\Var$.  This extension
$\lift\sstate\colon\EE\Var\to\EE\Var$ is
defined by
$\app{\lift\sstate}\x=\app\sstate\x$ and
$\app{\lift\sstate}{\operator(\exp_1,\dots,\exp_n)} =
\operator(\app{\lift\sstate}{\exp_1},\dots,\app{\lift\sstate}{\exp_n})$.
Now $\lift\sstate = \mu_\Var\circ\EE\sstate$, where
$\EE\sstate \colon \EE\Var \to \EE^2(\Var)$ performs uniform
substitution of variables by expressions, and
$\mu_\Var \colon \EE^2(\Var) \to \EE\Var$---seemingly the identity
function---glues expressions together.  Contrast this with the
inductive extension $\lift\state$ of a concrete state
$\state\colon\Var\to\Val$ for which
$\lift\state=\lift\evaluate\circ\EE\state$: one evaluates expressions
of expressions as a new expression, the other evaluates expressions of
values as a value.
We sometimes write $\sstate$ for~$\lift\sstate$.

Symbolic states as substitutions can be applied to predicates in $\PP\Var$ by letting $\app{\sstate}{\predicate(\exp_1,\dots,\exp_n)}:=\predicate(\app\sstate{\exp_1},\dots,\app\sstate{\exp_n})$.
Similarly, they can be applied recursively on Boolean expressions with predicates as base cases.
\begin{definition}[Symbolic States as Concrete State Transformers]\label{def:substeval}
  Let $\state \in \Val^\Var$ be a concrete state and $\sstate \in (\EE\Var)^\Var$ a symbolic state.
  Then \emph{$\state$ after $\sstate$} is the new concrete state $\state\after\sstate := \lift\state \circ \sstate \colon \Var \xrightarrow {\sstate} \EE\Var \xrightarrow {\EE \state} \EE\Val \xrightarrow {\evaluate} \Val $.
\end{definition}
\noindent 
The new state $\state\after\sstate$ evaluates an expression~$\exp$ by inductive extension, but we can also first apply $\sstate$ and then evaluate expression $\app\sstate\exp$ in the initial state~$\state$.
These two always coincide: evaluating $\exp$ in $\state\after\sstate$ is equivalent to evaluating $\app\sstate\exp$ in the initial {state~$\state$}:
\begin{restatable}[Substitution Lemma]{lemma}{substlemma} \label{lem:substlemma}
 Let $\state \in \Val^\Var$ and $\sstate \in (\EE\Var)^\Var$. Then 
 \begin{enumerate*}[label=(\roman*)]
   \item for all expressions $\exp\in\EE\Var$, $(\state\after\sstate)(\exp) = \state(\sstate(\exp))$; and
   \item for all Boolean expressions $\bexp\in\BB\Var$, $\state\after\sstate \sat \bexp$ iff $\state \sat \app\sstate\bexp$.
 \end{enumerate*}
\end{restatable}
\begin{example}[Symbolic states as concrete state transformers]\label{ex:statetransformers}
 Consider two variables $\x,\y$, a symbolic state ${\sstate =[\x \mapsto 2 * \x, \y \mapsto \x]}$ and ${\state_0 = [\x \mapsto 21, \y \mapsto 0]}$ a concrete state.
 Then $(\state_0\after\sstate)(\x) = \evaluate(\EE\state(2 * \x)) =
 \evaluate(2 * 21) = 42$ and similarly $(\state_0\after\sstate)(\y) = 21$.
 In general, for this $\sstate$, the map $\state \mapsto \state\after\sstate$ is a concrete state transformer that doubles the value of $\x$ and sets $\y$ equal to the old value of $\x$.
\end{example}

\noindent Symbolic semantics is derived by interpreting meta-variables as symbolic states:
\begin{definition}[Meta-Substitution of Symbolic States] \label{def:metasubstsymbolic} A \emph{meta-substitution of symbolic states} is a map $\symmetasubst_S \colon \SVar \to (\EE\Var)^\Var$.
\end{definition}
The way $\symmetasubst_S$ acts on the conclusion output $\mc \colon \Var \to \EE(\SVar \times \Var)$ of a symbolic SOS rule is analogous to meta-substitution of concrete states:
$$ \symmetasubst_S(\mc) \colon \Var \to \EE((\EE\Var)^\Var\times\Var) \xrightarrow{\EE(\eval)} \EE^2(\Var) \xrightarrow{\mu_\Var} \EE\Var $$
The meta-substitution $\symmetasubst_S$ is first universally applied to $\mc$, then every pair $(\sstate,\x)$ is evaluated within the expressions, and finally, the resulting expression is glued.

\begin{example}[Meta-substitution for assignment] \label{ex:metasubstassign} Let $\state_0=[\x\mapsto 21,\y\mapsto 0]$ and $\sstate=[\x\mapsto2*\x,\y\mapsto\x]$, and let $\state = [\x\mapsto 42,\y\mapsto 21]$.
 In \Cref{ex:statetransformers}, we saw that $\sstate$ transforms $\state_0$ to $\state$, i.e., $\state_0\after\sstate = \state$.
 Suppose $\metasubst_S$ and $\symmetasubst_S$ are concrete and symbolic meta-substitutions such that $\metasubst_S(\ma)=\state$ and $\symmetasubst_S(\ma) = \sstate$.
 Consider an assignment $\x\asgn 0-\x$ as in \Cref{ex:concretestep} and its transition axiom $(\x\asgn 0-\x,\ma) \metatrterm \mc$ in the specification from \Cref{ex:specification}.
 With $\Var=\{\x,\y\}$, we have $\mc \colon \Var \to \EE(\{\ma\}\times\Var)$ such that $\mc \colon \x \mapsto 0-\ma(\x)$ and $\mc \colon \y \mapsto \ma(\y)$.
 Putting $\state' := \metasubst_S(\mc)$ and $\sstate' := \symmetasubst_S(\mc)$, we have $\state'(\x) = 0-\state(\x) = -42$ and $\sstate'(\x) = 0-\sstate(\x) = 0-2*\x$.
 For $\y$, $\state'(\y)=\state(\y) = 21$ and $\sstate'(\y) = \sstate(\y) = \x$.
 Therefore, $\state_0 \after \sstate' = \state'$.
\end{example}

In this example, the concrete and symbolic states are transformed concertedly by the assignment update.
Specifically, $\state_0 \after \symmetasubst_S(\mc) = \metasubst_S(\mc)$ follows from ${\state_0 \after \symmetasubst_S(\ma)} = \metasubst_S(\ma)$ because $\ma$ is the only meta-variable occurring in $\mc$.
Thus, the example illustrates a general inductive property of our rule format: at every step, the change in symbolic state matches the change in concrete state.

\begin{restatable}[Meta-Substitution Lemma]{lemma}{metasubstlemma}\label{lem:metasubstlemma}
 Let $\metasubst_S \colon \SVar \to \Val^\Var$ be a concrete and $\symmetasubst_S \colon \SVar \to (\EE\Var)^\Var$ a symbolic meta-substitution and let $\state_0 \in \Val^\Var$ be a concrete state.
 For all $\mc \colon \Var \to \EE(\SVar\times\Var)$, if $\state_0\after\symmetasubst_S(\mb) = \metasubst_S(\mb)$ for all $\mb \in \SVar$ that occur in $\mc$, then $\state_0\after\symmetasubst_S(\mc) = \metasubst_S(\mc)$.
\end{restatable}
\noindent A \emph{symbolic execution model} is a nondeterministic unlabeled transition system $(\Prg\times(\EE\Var)^\Var\times\BB\Var,\symbextr)$.
We now define the symbolic execution model $\symbextr$ \emph{intended} by a symbolic SOS specification $\symbexspec$ for the signature $\signature$.
For this, we let $n=\arity(\operator)$ and, given arbitrary state $\sstate$ and path condition $\pc$, we recursively define the set of outgoing transitions for $(\operator(\p_1,\dots,\p_n),\sstate,\pc)$, where we have already defined the sets $\P_i = \{ (\p',\sstate',\pc') \mid (\p_i,\sstate,\pc) \symbextr (\p',\sstate',\pc') \}$ of outgoing transitions for the subterms $\p_i$.
An $n$-tuple $\xi \in \prod_{i\in[1..n]} \P_i$ contains one possible combination of targets for $\p_i,\dots,\p_n$.
Let $W_\xi \subseteq \{1,\dots,n\}$ be the set of~$i$ with $\p^{(i)}_\xi \neq \terminated$, indicating progressive premises; write ${\xi = (\p^{(i)}_\xi,\sstate^{(i)}_\xi,\pc^{(i)}_\xi)_{i\in[1..n]}}$.
\begin{definition}[Symbolic Semantics] \label{def:symbolicsemantics}
  Let $\operator\in\Sigma$ be an operator, $\sstate \in (\EE\Var)^\Var$ a symbolic state, $\pc\in\BB\Var$ a path condition, and ${\p_1,\dots,\p_n \in \Prg}$ a set of programs.
 For $\xi\in\prod_{i\in[1..n]}\P_i$, let~$\RR_{\xi,1}\in\symbexspec$
 and~$\RR_{\xi,2}\in\symbexspec$ be the rules for $\operator$ and $W_\xi$ with triggers $\bexp_\operator$ and $\neg\bexp_\operator$, respectively.
 Let $(\operator(\p_1,\dots,\p_n),\sstate,\pc) \symbextr (\p',\sstate',\pc')$, by definition, for all $(\p',\sstate',\pc')$ in the set
 $$\big\{\,\,(\symmetasubst_\xi(\mt_{\xi,1}),\, \symmetasubst_\xi(\mc_{\xi,1}), \pc \land \app\sstate\bexp \land \PC_\xi), \,\, (\symmetasubst_\xi(\mt_{\xi,2}),\, \symmetasubst_\xi(\mc_{\xi,2}),\, \pc \land \neg\app\sstate\bexp \land \PC_\xi)\,\,\big\}_{\xi\in\prod_{i} \P_i}$$
 where $\mt_{\xi,j},\mc_{\xi,j}$ ($j=1,2$) are the conclusion targets and outputs of $\RR_{\xi,j}$, 
 $$ \symmetasubst_\xi \colon \XVar+\SVar \to \Prg + (\EE\Var)^\Var,\,\,\, \mx \mapsto 
  \begin{cases} 
    \p_i & \textup{if }\mx=\mx_i\\
    \p_\xi^{(i)} & \textup{if }\mx=\mx'_i
  \end{cases} 
  \quad \mb \mapsto 
  \begin{cases} 
    \sstate & \textup{if }\mb=\ma\\ 
    \sstate_\xi^{(i)} & \textup{if }\mb=\mb_i 
  \end{cases}
  $$ and $\PC_\xi := \bigwedge_{i\in[1..n]}\pc_\xi^{(i)}$ is the conjunction of all path conditions in the premises.
\end{definition}
Here, $\PC_\xi$ includes \emph{all} premise path conditions, including conditions potentially not used in the conclusion.
This condition may appear too strong for some rule instances in symbolic execution techniques.
However, since $\prod_{i\in[1..n]}\P_i$ comprises all combinations of transitions, and since every $\xi$ induces a step for both $\app\sstate\bexp$ and $\neg\app\sstate\bexp$, the resulting set of path conditions covers all of the input path condition~$\pc$.
Many of the resulting steps may have coinciding continuations.
\begin{restatable}[Path Condition One-Step Coverage]{proposition}{coverageprop} \label{prop:pccoverage}
  Let $\symbextr$ be the intended symbolic execution system of a symbolic SOS specification, and let $(\p,\sstate,\pc)$ be a symbolic configuration.
  For $A = \{ \pc' \mid (\p,\sstate,\pc) \symbextr (\p',\sstate',\pc') \}$:
  \begin{itemize}
    \item $\pc_1 \wedge \pc_2 \equiv \false$ for all $\pc_1,\pc_2 \in A$ such that $\pc_1 \neq \pc_2$; and 
    \item $\bigvee A \equiv \pc$, where $\bigvee A$ denotes finite disjunction of all elements in $A$.
  \end{itemize}
\end{restatable}
\begin{example}\label{ex:symbolicstep}
  We return to program $\pabs$ from \Cref{ex:concretestep} and the symbolic SOS specification from \Cref{ex:specification}.
  Using the two axioms for an \emph{if} statement, the derived symbolic execution semantics gives the two transitions $(\pabs,\id, \top) \symbextr (\x\asgn-\x, \id, \top\land(\x<0))$ and $(\pabs,\id, \top) \symbextr (\pskip, \id, \top\land\neg(\x<0))$.
  Continuing for one more step we obtain a set of four reachable configurations; the two on the left stem from $\x\asgn-\x$; the other two from $\pskip$:
  $$\left\{
    \begin{aligned}
      (\terminated,& \x\mapsto-\x, \top\land(\x<0)\land\top),\, && (\terminated, \id, \top\land\neg(\x<0)\land\top), \\
      (\terminated,& \x\mapsto-\x, \top\land(\x<0)\land\neg\top),\, && (\terminated, \id, \top\land\neg(\x<0)\land\neg\top)
    \end{aligned}
  \right\}
  $$
  The bottom configurations can be discarded: no state satisfies $\neg\top$.
  We further simplify path conditions by removing $\top$-conjuncts.
  Then $\symbextr$ reduces the \emph{if} statement to two possible transformations: $\x\mapsto-\x$ if $\x<0$ and $\id$ otherwise.
\end{example}

\noindent
\emph{Syncrete bisimulation} is a coinductive formalization of correctness and completeness for symbolic execution.
We prove that our rule format induces a syncrete bisimulation relation between concrete and symbolic execution semantics, namely the identity relation: every program is syncretely bisimilar to itself.
\begin{definition}[Syncrete Bisimulation] \label{def:syncretebisimulation}
  Let $\consystem$ be a concrete execution model and $\symbexsystem$ a symbolic execution model.
  A relation $R \subseteq \Prg\times \Prg$ is a \emph{syncrete bisimulation between $\contr$ and $\symbextr$} if, for all $\sstate\in (\EE\Var)^\Var$, $\pc\in\BB\Var$, and initial states $\state_0 \in \Val^\Var$ s.t.\ $\state_0\sat\pc$, whenever $\p R\q$:
  \begin{itemize}
    \item if $(\p,\state_0\after\sstate) \contr (\p',\state')$ then there is $(\q',\sstate',\pc')$ such that\\
      \rom{1} $(\q,\sstate,\pc) \symbextr (\q',\sstate',\pc')$
      \;\rom{2} $\state'=\state_0\after\sstate'$
      \;\rom{3} $\p'R\q'$ and \; \rom{4} $\state_0\sat\pc'$.
    \item if $(\q,\sstate,\pc) \symbextr (\q',\sstate',\pc')$ and $\state_0\sat\pc'$ then $(\p,\state_0\after\sstate) \contr (\p',\state_0\after\sstate')$ for some $\p'\in \Prg$ such that $\p'R\q'$.
  \end{itemize}
\end{definition}
The first item makes every step in the symbolic system coinductively complete: every concrete step is matched by a symbolic step that refines the path condition in a way that the initial state $\state_0$
remains satisfied.
The second item makes every step in the symbolic system coinductively correct.
Here, it may seem like any symbolic state $\sstate$ can be chosen, but the updated path condition $\pc'$ always contains the Boolean formula $\bexp$ that guards control-flow under substitution by~$\sstate$.
Hence, the condition $\state_0 \sat \pc'$ entails $\state_0 \after \sstate \sat \bexp$.

In the following results, let $\symbexspec$ be a symbolic SOS specification and let $\contr$ be the intended concrete model and $\symbextr$ the intended symbolic model.
\begin{restatable}{theorem}{syncretebisim}\label{thm:syncretebisim}
  The identity relation on the set $\Prg$ of programs is a syncrete bisimulation between $\contr$ and $\symbextr$.
\end{restatable}
\noindent By induction on the length of the transition chain, with the definitions of correctness and completeness from \Cref{sec:overview}:
\begin{restatable}[Correctness and Completeness]{corollary}{correctnessandcompleteness}\label{cor:correctness}
  The intended model of symbolic execution $\symbextr$ is correct and complete with respect to $\contr$.
\end{restatable}\noindent
The induced small-step model provides a correct and complete core for symbolic execution.
Full symbolic execution amounts to providing a search strategy for the execution tree built by $\symbextr$, and the result is guaranteed to correspond to concrete program behavior by \Cref{cor:correctness}.

\section{Extensions}\label{sec:extensions}
We consider two extensions to
\While:
\emph{arrays}
(see De Boer and
Bonsangue~\cite{deboerSymbolicExecutionFormally2021}) and a
\emph{probabilistic} programming constructs
(see
Voogd et al.~\cite{voogd2023symprob}).  We immediately obtain
a concrete execution model $\consystem$ and a symbolic execution
model $\symbextr$ that is both correct and complete with respect to
$\contr$.

\paragraph{Arrays.}
Arrays can be incorporated in \While by imposing some structure on $\Var$, $\ExpOp$ and $\LogOp$.
Let the \emph{variables} be a disjoint union of regular and array variables $\Var+\AVar$ with values in $\Z + (\N\partial\Z)$.
Regular variables $\x\in\Var$ are assigned integers and $\a\in\AVar$ partial integer-valued functions with index domain $\N$.

Let \emph{expressions} include $\Aexp$, $\Aasgn$ and $\len$ for $\a\in\AVar$ and $\exp,\exp'\in\EE\Var$.
The semantics is modeled by $\evaluate$; we let $\evaluate$ map $\Aexp$ to $\evaluate(\a)(\evaluate(\exp))$, and $\Aasgn(\exp'')$ to $ \evaluate(\exp')$ if $\evaluate(\exp)=\evaluate(\exp'')$ or $\evaluate(\a[\exp''])$ otherwise.
We let $\evaluate(\len)$ be the size of the set on which $\evaluate(\a)$ is defined.
Now extend $\Sigma$ from \Cref{ex:grammar} by allowing the left-hand side of an assignment to be an array expression and introduce a new constant $\error$ to denote out-of-bounds access.

Let $\oob \in \PredOp$ be a unary predicate indicating absence of indexing errors.
The semantics of the closed predicate $\oob(\exp)$ is inductively defined by \rom 1 $\oob(n) := \truth$ for all constants $n\in\ExpOp$; \rom 2 $\oob(\Aexp)$ iff $(0\le\exp<\len) \land \oob(\exp)$; \rom 3 $\oob(\Aasgn)$ iff $(0\le\exp<\len) \land \oob(\exp) \land \oob(\exp')$; and \rom 4 $\oob(\operator(\exp_{1},\dots,\exp_{n}))\equiv\oob(\exp_{1})\land\dots\land\oob(\exp_{n})$
for $n$-ary operation symbols $\operator\in\ExpOp$.

Now let us define a symbolic SOS specification for signature $(\Sigma,\ExpOp,\PredOp,\LogOp,\arity,\guardsign)$.
Array assignments require a new pair of rules:
$$\begin{tabular}{C{.5\textwidth}C{.5\textwidth}}
  $\srule{}{\oob(\Aexp)}{(\Aexp\asgn\exp',\state) \metatrterm \state[\a\mapsto\a[\state(\exp) := \state(\exp')]]}$ &
  $\srule{}{\neg\oob(\Aexp)}{(\Aexp\asgn\exp',\state) \metatr (\error, \state)}$
\end{tabular}
$$
The conclusion of the left rule denotes the map $\Var+\AVar \to \EE(\{\ms\}\times(\Var+\AVar))$ that maps each variable $\x$ to $\state(\x)$ (including arrays) except that $\a$ is mapped to the \emph{expression} $\a[\state(\exp) := \state(\exp')]$
with $\exp, \exp'$ updated to replace each variable $\y\in\Var+\AVar$
with $\state(\y)$.
The rule on the right signals an error that may be handled by other mechanisms.

The rules in \Cref{fig:whilespec} with trivial triggers can be safely
replaced with two rules: one with the additional trigger $\oob(\exp)$ for
expressions in the program term and one with $\neg\oob(\exp)$ progressing to an error.
For \Rule{if-} and \Rule{while-} rules some additional machinery is needed to maintain the requirements of \Cref{def:symsos}.
For each Boolean expression $\bexp$ we introduce a new binary operator $\psif{\bexp}{\,\cdot\,}{\,\cdot\,}$ (for ``safe if'').
We then have exactly two rules for each of $\syntax{if}~\bexp$ and $\syntax{sif}~\bexp$.

$$\begin{tabular}{C{.5\textwidth}C{.5\textwidth}}
 \srule{if-safe}{\oob(\bexp)}{(\pif\bexp\mx\my,\ma) \metatr (\psif\bexp\mx\my,\ma)} &
 \srule{if-err}{\neg\oob(\bexp)}{(\pif\bexp\mx\my,\ma) \metatr (\error, \ma)} \\\\
 \srule{if-T}{\bexp}{(\psif\bexp\mx\my,\ma) \metatr (\mx,\ma)} &
 \srule{if-F}{\neg\bexp}{(\psif\bexp\mx\my,\ma) \metatr (\my,\ma)} \\
\end{tabular}
$$

Thus an if-statement first checks if its condition contains a nil error, and only then proceeds safely to one of its branches.
A ``safe while'' is implemented analogously by first checking its conditition, and then proceeding as before.

\paragraph{Randomization.}
For probabilistic sampling during program execution, we consider a set of logical variables $\mathcal Y = \{\y_k\}_{k\in\N}$ that represent samples; states are now maps $\Var+\mathcal Y \to \Val$.
We consider a signature $\signature$ similar to the While language (see \Cref{ex:whilesignature}).
To ensure probabilistic independence of samples, assignments $\assign$ cannot involve variables from $\mathcal Y$; they are still represented by a pair $(\x,\exp)\in\Var\times\EE\Var$, but $\Sigma$ is extended with a sampling statement $\sample$.
Consider the rule for sampling (with guard $\true$):
$$ \saxiom{}{\sample, \ma \metatrterm \{\x \mapsto (\ma,\y_0), \y_0 \mapsto (\ma,\y_1), \y_1 \mapsto (\ma,\y_2), \dots \} }$$
which stores the first available sample $\y_0$ in $\x$ and shifts all other samples one position.
Writing $(\state,\samplestream)$ for the state $\Var + \mathcal Y \to \Val$, the concrete rule is
$$ \saxiom{}{\sample, \state, \samplestream \metatrterm (\state[\x \leftarrow \samplestream_0],\tail(\samplestream))}$$
Taking the head and tail does not work in the symbolic counterpart of this rule.
A solution is to introduce a sampling index $k$ and using the rule 
$$\saxiom{}{\sample, \sstate, k \metatrterm \sstate[\x \mapsto \y_k], k+1}$$
An indexing scheme like this must be proven correct in the presence of all the rules in the programming system.
The symbolic model of this language system produces symbolic states $\sstate :\Var + \mathcal Y \to \EE(\Var + \mathcal Y)$.
Keeping $k$ constant in other rules ensures that the part $\mathcal Y \to \EE(\Var + \mathcal Y)$ left-shifts the stream by $k$, always giving a new variable in~$\mathcal Y$.

For this to be a true randomization of programs, one assumes that the values for $\mathcal Y$, given by the map $\samplestream : \mathcal Y \to \Val$, adhere to some probability law.
This is a modeling issue; we argue that the symbolic system produces symbolic states that accurately represent program behavior in that they produce the same result once this initial state (randomized or not) is evaluated by the symbolic state.
For a detailed account of symbolic execution of probabilistic programs, see \cite{voogd2023symprob}.

\section{Related Work}
De Boer and Bonsangue formalize symbolic execution for the \While language and define the notions of correctness and completeness of symbolic execution~\cite{deboerSymbolicExecutionFormally2021}.
They employ a small-step transition system and inductive proofs of correctness and completeness over its transitive closure.
In contrast, our work offers a coinductive alternative (allowing for non-finite computations) that captures both correctness and completeness in terms of syncrete bisimulation.
We do this by quantifying over conceptual initial states $\state_0$, and making concrete small-steps on $\state_{0}\after \sstate$ rather than big-steps on $\state_0$ itself.
A symbolic reconfiguration from $\sstate$ to $\sstate'$ then corresponds to a concrete reconfiguration of $\state_{0}\after\sstate$ to $\state_{0}\after\sstate'$.

Goncharov et al. developed \emph{stateful SOS} (SSOS) for stateful programs~\cite{goncharov2022stateful}, extending 
 GSOS~\cite{bloom1995bisim} with \emph{state}, focusing crucially on the compositionality of the derived semantics.
Via a reduction to GSOS, SSOS specifications are shown to correspond precisely to natural transformations which induce a denotational behavior that ensures compositionality in \emph{resumption} semantics, a very fine-grained semantics in which very few programs are considered equivalent~\cite{goncharov2022stateful}.
In particular, programs that induce the same state transformation --- like $\seq{\x\asgn 1}{\x\asgn\x+1}$ and $\seq{\x\asgn 1}{\x\asgn\x*2}$ ---
may not be equivalent under resumption semantics.
Goncharov et al.\,consider two coarser semantics --- trace and termination --- which fail to be compositional in general.
They therefore propose further restrictions on the SSOS format to ensure compositionality also in these settings.
The unrestricted SSOS format forms the basis for symbolic SOS in our work, but we refine state transformations to ensure that they can be symbolically simulated.

The $\K$ framework shares our goal of defining language semantics with correct and complete analysis tools~\cite{rosu-2017-marktoberdorf}.
In particular, Lucanu et al. develop a language-independent coinductive description of symbolic execution~\cite{arusoaie2013generic,lucanu2017generic}, based on \emph{Reachability Logic}~\cite{stefanescu2014reachability}.
They use matching logic and reachability logic to define semantics as rewrite rules, whereas we provide syntactic restrictions on the common stateful SOS format and introduce syncrete bisimulation as a useful formalization of the correspondence between symbolic and concrete (small-step) semantics.
As a consequence, our proofs mostly use structural induction over programs whereas their proofs use correspondences between proof trees.

Bodin et al. propose another language-independent framework that provides analysis tools ``for free'' with their pretty-big-step semantics~\cite{bodin2015pbs} and Skeletal Semantics~\cite{bodin2019skeletal}.
They provide a framework of simple building blocks (bones) that assemble into programs (skeletons).
Skeletons are given interpretations, and generic consistency results between interpretations are established.
Finally, they define concrete and abstract interpretations and instantiate the consistency results with language-dependent lemmas.
Their approach differs from ours by focusing on \emph{structural}
building blocks of semantics rather than a rule format.
Additionally they focus on abstract interpretation rather than symbolic execution.

\section{Discussion}
We briefly consider two interesting aspects of the presented work: (1) the conditions on the rule formats to simultaneously construct concrete and symbolic semantics, and (2) extensions to more low-level state representations.
  
Our symbolic SOS format provides a \emph{sufficient} condition for both concrete and symbolic semantics to be constructed simultaneously.
However, identifying \emph{necessary} conditions for rule formats that ensure correctness and completeness would be very challenging, because a lot of design choices have to be made to bridge the gap between the desired properties and the semantics specification.
Our  work on symbolic SOS is based on the following important design choices:
\begin{itemize}
  \item Symbolic SOS builds on GSOS, which ensures that bisimilarity is always a congruence and
    that canonical operational models exist.  GSOS provides a
    sufficient condition for this property. GSOS seems fairly close to
    the limits of well-behaved SOS formats. (For more liberal formats
    such as ntyft/ntyxt which allows both look-ahead and negative premises~\cite{DBLP:journals/tcs/Groote93},
	the interperation is more subtle~\cite{DBLP:journals/jlp/Glabbeek04} and it can even be difficult to decide whether a
    (unique) model exists~\cite{DBLP:journals/jlp/KlinN17}.)
  \item Symbolic SOS builds on stateful SOS, which ensures that the
    properties above hold
    in a stateful setting.  Symbolic
    SOS imposes structure such that states can be interpreted both
    concretely and symbolically.
    \Cref{lem:metasubstlemma} (meta-substitution) proves that every
    step in one system is simulated by a step in the other.  Rules in
    a specification must come in pairs --- generalizable to
    complementary tuples of arbitrary size --- with mutually disjoint
    conditions.
\end{itemize}
  For our techniques to apply to a language,
  its operational semantics must be expressible with
  rules that
  syntactically enforce these
  properties:
  GSOS-like restrictions for compositionality and our added
  requirement on state structure.

  The sets of variables and values, and the signatures of expressions
  and Boolean expressions have intentionally been kept abstract.
  \Cref{sec:extensions}
  shows how symbolic execution correctness and completeness is
  maintained with additional structuring of the signatures.
  We believe
  that other extensions, e.g., for pointers or aliasing, would work
  similarly.
  Heaps could then be implemented by imposing structure on the
  states (both concrete and symbolic) similar to
  the arrays of \Cref{sec:extensions}.  By distinguishing non-heap and
  heap variables and evaluating heap variables in partial maps (like
  the array variables), a pointer map is emulated.  A predicate
  similar to the absence of indexing errors can be used to detect null
  pointer exceptions.  See \cite{deboerSymbolicExecutionFormally2021}
  for a discussion on aliasing. In this context, it would be
  interesting
  to further extend the path conditions with separation logic for
  pointers \cite{berdine2005symbolic}.  This would not affect the
  results presented in this paper, provided the meta-substitution
  lemma (\Cref{lem:metasubstlemma}) is maintained, ensuring
  that
  the symbolic SOS rules define matching transitions for the concrete
  and symbolic semantics.

\section{Conclusion}
We present a language-independent rule format for program semantics that induces both concrete and symbolic models.
The induced models enjoy a bisimilarity relationship that ensures correctness and completeness of symbolic execution.
Our approach thus allows to define operational semantics for a language and immediately obtain a symbolic execution engine that
is correct by construction, providing a formal basis for analysis and verification tools.

Technically, we exploit that symbolic states represent transformations of concrete states to augment stateful SOS with a more structured notion of state, thereby obtaining \emph{symbolic} stateful SOS\@.
From symbolic SOS specifications, we show how to derive execution models in terms of symbolic and concrete transition systems.
We formulate the novel notion of \emph{syncrete bisimulation} and use it to prove that the derived symbolic execution model is correct and complete.
The proof makes crucial use of the notion of bisimulation and a very general substitution lemma that relates symbolic states and concrete state transformations.

Our results work for concrete semantics that can be understood as
\emph{deterministic} state transformers.  An interesting direction of
development would be to generalize this to nondeterministic settings,
such as concurrent programs.  This would require a deeper
investigation of the natural transformations induced by the symbolic
SOS rule format and their categorical semantics.  Goncharov et
al.~\cite{goncharov2022stateful} also highlight this direction of
research for the non-symbolic case.

\paragraph{Acknowledgements}
The authors would like to thank the anonymous reviewers for their insightful questions and feedback.

\bibliographystyle{splncs04}
\bibliography{ref.bib}

\appendix 
\section{Proofs}
In the proofs of \Cref{lem:substlemma,lem:metasubstlemma}, we will use some universal algebra.
The family $\evaluate$ of maps $(\evaluate_\operator \colon \Val^{\arity(\operator)} \to \Val)_{\operator\in\ExpOp}$ is an $\ExpOp$-algebra structure for $\Val$, making $(\Val,\evaluate)$ an $\ExpOp$-algebra.
Recall that $\EE$ generates terms using operators in $\ExpOp$ and note that $(\EE,\mu,\eta)$ is a monad, with multiplication $\mu_X\colon\EE^2(X)\to\EE(X)$ gluing together expressions and $\eta_X \colon X \to \EE(X)$ interpreting variables as expressions.

We have defined $\lift\evaluate\colon\EE(\Val)\to\Val$ as the inductive extension of $\evaluate$ over itself, i.e., as the unique map ($\ExpOp$-algebra homomorphism) such that $\evaluate\circ \ExpOp(\lift\evaluate) = \lift\evaluate\circ\iota$, where $\iota\colon\ExpOp(\EE(\Val))\to\EE(\Val)$ ``injects'' the expression $\operator(\exp_1,\dots,\exp_n)$ into $\EE\Val$.
Note that $\lift\evaluate$ is an algebra for the monad $\EE$, meaning $\lift\evaluate\circ\EE(\lift\evaluate) = \lift\evaluate\circ\mu_\Val$.
Indeed, we can evaluate sub-expressions ($\EE(\lift\evaluate)$) and then evaluate the obtained expression ($\lift\evaluate$), or, equivalently, we can glue it as one big expression ($\mu_\Val$) and then evaluate the whole expression ($\lift\evaluate$).
This fact is true ``by induction''.
Rather, it \emph{is} induction: both ways evaluate expressions from the bottom up.

\subsection{Proof of \Cref{lem:substlemma}}
\substlemma*
\begin{proof}
  In this proof we will always use the liftings $\lift\state=\lift\evaluate\circ\EE\state$ of $\state$ and $\lift\sstate=\mu_\Var\circ\EE\sstate$ of $\sstate$ explicitly, so we prove $\lift{(\state\after\sstate)}(\exp) = \lift\state(\app{\lift\sstate}\exp)$ for all expressions $\exp\in\EE\Var$.

  The first item is proved by induction on the structure of terms in $\EE\Var$.
  It is true by definition for variables: 
  $$\lift{(\state\after\sstate)}(\x) = (\state\after\sstate)(\x)= \lift\state(\app\sstate\x) = \lift\state(\app{\lift\sstate}\x)$$
  Extending to expressions, note that the lifting of $\state\after\sstate$ is $\lift\evaluate \circ \EE(\state\after\sstate)$, i.e.,
  \begin{align}\label{eq:evalsubst1}
    \lift{(\state\after\sstate)} \colon \EE\Var \xrightarrow{\EE\sstate} \EE^2\Var \xrightarrow{\EE^2\state} \EE^2\Val \xrightarrow{\EE(\lift\evaluate)} \EE\Val \xrightarrow{\lift\evaluate} \Val 
  \end{align}
  On the other hand,
  \begin{align}\label{eq:evalsubst2}
    \lift\state\circ\lift\sstate \colon \EE\Var \xrightarrow{\EE\sstate} \EE^2\Var \xrightarrow{\mu_\Var} \EE\Var \xrightarrow {\EE\state} \EE\Val \xrightarrow {\lift\evaluate} \Val 
  \end{align}
  The goal is to show that \eqref{eq:evalsubst1} and \eqref{eq:evalsubst2} coincide.
  The diagrams
  $$ \xymatrix{
 *+{\EE^2\Var} \ar[r]^{\EE^2\state} \ar[d]^{\mu_\Var} & *+{\EE^2\Val} \ar[d]^{\mu_\Val} \\
 *+{\EE\Var} \ar[r]_{\EE\state} & *+{\EE\Val} 
}\qquad \xymatrix{
    *+{\EE^2\Val} \ar[r]^{\EE{\lift\evaluate}} \ar[d]^{\mu_\Val}
 & *+{\EE\Val} \ar[d]^{{\lift\evaluate}} \\
 *+{\EE\Val} \ar[r]_{{\lift\evaluate}}
 & *+{\Val}
} $$
  commute, because multiplication is natural and $\lift\evaluate$ is an $\EE$-algebra structure.
  Thus,
  $$ \xymatrix{
    *+{\EE\Var} \ar[r]^{\EE\sstate}
 & *+{\EE^2\Var} \ar[r]^{\EE^2\state} \ar[d]^{\mu_\Var}
 & *+{\EE^2\Val} \ar[r]^{\EE\lift\evaluate} \ar[d]^{\mu_\Val}
 & *+{\EE\Val} \ar[d]^{\lift\evaluate} \\
 & *+{\EE\Var} \ar[r]_{\EE\state}
 & *+{\EE\Val} \ar[r]_{\lift\evaluate}
 & *+{\Val}
}$$
  is a commuting diagram, meaning that \eqref{eq:evalsubst1} and \eqref{eq:evalsubst2} indeed coincide.

  For the second item, we have for all $n$-ary predicates $\predicate \in \PP\Var$,
  $$\begin{array}{rl}
    \state\after\sstate\sat\predicate(\exp_1,\dots,\exp_n) 
    \iff
    & \interpret(\predicate(\lift{\state\after\sstate}(\exp_1),\dots,\lift{\state\after\sstate}(\exp_n)))=\truth \\
    \iff
    & \interpret(\predicate(\lift\state(\app{\lift\sstate}{\exp_1}),\dots,\lift\state(\app{\lift\sstate}{\exp_n})))=\truth \\
    \iff 
    & \state\sat\predicate(\app{\lift\sstate}{\exp_1},\dots,\app{\lift\sstate}{\exp_1}) \\
    \iff 
    & \state\sat\app\sstate{\predicate(\exp_1,\dots,\exp_n)}
  \end{array}$$
  and this constitutes the base case for Boolean expressions.
  For $\logop\in\LogOp$ an $n$-ary logical operator with truth table $t_\logop \colon 2^n \to 2$ (where $2=\{\truth,\falsehood\}$), note that $\state\after\sstate\sat\logop(\bexp_1,\dots,\bexp_n)$ if and only if $t_\logop(\delta_1,\dots,\delta_n)=\truth$ where each $\delta_i = \truth$ if $\state\after\sstate\sat\bexp_i$ and $\delta_i=\falsehood$ otherwise.
  By IHs, $\state\after\sstate\sat\bexp_i$ iff $\state\sat\app\sstate{\bexp_i}$ and so so $\state\after\sstate\sat\logop(\bexp_1,\dots,\bexp_n)$ iff $\state\sat\app\sstate{\logop(\bexp_1,\dots,\bexp_n)}$.
    By induction, then $\state\after\sstate\sat\bexp$ iff $\state\sat\app\sstate\bexp$ for all $\bexp \in \BB\Var$.
\end{proof}

\subsection{Proof of \Cref{lem:metasubstlemma}}
\metasubstlemma*
\begin{proof}
  Without loss of generalization, $\state \after \symmetasubst_S(\ma) = \metasubst_S(\ma)$ for \emph{all} $\ma \in \SVar$.
  Let $\metasubst_S' \colon \SVar\times\Var \to \Val$ and $\symmetasubst_S'\colon \SVar\times\Var \to \EE\Var$ be the curried versions of $\metasubst_S$ and $\symmetasubst_S$.
  Then 
  $$ \metasubst_S' = \lift\evaluate \circ \EE\state \circ \symmetasubst_S' $$
  Applying $\EE$ to this identity and evaluating using $\lift\evaluate$ it follows that $$ \begin{array}{rlr}
    \lift\evaluate \circ \EE\metasubst_S = & \lift\evaluate \circ \EE\lift\evaluate \circ \EE^2 \state \circ \EE\symmetasubst_S & \\
    = & \lift\evaluate \circ \mu_\Val \circ \EE^2\state \circ \EE\symmetasubst_S & \qquad \qquad \lift\evaluate \textup{ is a monad algebra} \\
    = & \lift\evaluate \circ \EE \state \circ \mu_\Var \circ \EE\symmetasubst_S & \mu \textup{ is natural}
  \end{array}
  $$
  The above identity pointwise yields the following commuting diagram:
  $$ \xymatrix{*++{\EE(W\times\Var)^\Var}\ar[r]^<<<<<<<{(\EE\metasubst_S)^\Var}\ar[d]_{(\EE\symmetasubst_S)^\Var} &*++{(\EE\Val)^\Var} \ar[r]^{\lift\evaluate^\Var} & *++{\Val^\Var}\\ *++{(\EE^2\Var)^\Var}\ar[r]_{(\mu_\Var)^\Var} &*++{(\EE\Var)^\Var} \ar[r]_{(\EE\state)^\Var} & *++{(\EE\Val)^\Var} \ar[u]_{\lift\evaluate^\Var}}$$
  Now let $\mc \in \EE(W\times\Var)^\Var$ be arbitrary.
  The path along the top gives a concrete state, call it $\state' := \metasubst_S(\mc)\in\Val^\Var$.
  Two steps through the path along the bottom gives us a symbolic state $\sstate := \symmetasubst_S(\mc) \in (\EE\Var)^\Var$.
  The diagram tells us that $\lift\evaluate^\Var((\EE\state)^\Var(\sstate)) = \state'$.
  More precisely, for every $\x \in \Var$, we have $\lift\evaluate(\EE\state(\app\sstate\x))=\state'(\x)$.
  Recalling \Cref{def:substeval}, we conclude that $\state'=\state\after\sstate$, which was to be shown.
\end{proof}

\subsection{Proof of \Cref{prop:pccoverage}}
\coverageprop*
\begin{proof}
  We prove the two items separately.
  \begin{itemize}
    \item By induction on the structure of $\p$.
      The base cases are constants in $\Sigma$; the set $A$ consists of exactly two elements, one contained in the conjunct $\app\sstate\bexp$, the other in $\neg\app\sstate\bexp$.
      These are always disjoint.

      For the inductive step, write $\p = \operator(\p_1,\dots,\p_n)$ and let $\bexp$ be the Boolean guard of $\operator$.
      We have $(\p,\sstate,\pc) \symbextr (\q_1,\sstate_1,\pc_1)$ and $(\p,\sstate,\pc) \symbextr (\q_2,\sstate_2,\pc_2)$ (some $\q_1,\sstate_1$ and $\q_2,\sstate_2$).
      These were built from two $n$-tuples $$\xi_1 = (\p^{(i)}_1,\sstate^{(i)}_1,\pc^{(i)}_1)_{i\in[1..n]}
      \quad \textup { and } \quad \xi_2 = (\p^{(i)}_2,\sstate^{(i)}_2,\pc^{(i)}_2)_{i\in[1..n]}$$
      such that $(\p_i,\sstate,\pc) \symbextr (\p^{(i)}_1,\sstate^{(i)}_1,\pc^{(i)}_1)$ and $(\p_i,\sstate,\pc) \symbextr (\p^{(i)}_2,\sstate^{(i)}_2,\pc^{(i)}_2)$, $i=1,\dots,n$.
      The two path conditions $\pc_j$ $(j=1,2)$ in $A$ are one of two forms: 
      $$ \pc_j = \pc \wedge \app\sstate\bexp \wedge \bigwedge_{i=1}^n \pc^{(i)}_j \qquad \qquad \pc_j = \pc \wedge \neg\app\sstate\bexp \wedge \bigwedge_{i=1}^n \pc^{(i)}_j$$
      If they are not of the same form then they are disjoint since $\app\sstate\bexp \land \neg\app\sstate\bexp \equiv \false$.
      Therefore, w.l.o.g., both are of the first form.
      Since $\pc_1 \neq \pc_2$, then, there must be $i$ such that $\pc^{(i)}_1 \neq \pc^{(i)}_2$.
      We have $(\p_i,\sstate,\pc) \symbextr (\p^{(i)}_1,\sstate^{(i)}_1,\pc^{(i)}_1)$ and $(\p_i,\sstate,\pc) \symbextr (\p^{(i)}_2,\sstate^{(i)}_2,\pc^{(i)}_2)$.
      By IH, $\pc^{(i)}_1$ and $\pc^{(i)}_2$ are disjoint, and so also $\pc_1 \wedge \pc_2 \equiv \false$.
    \item By induction on the structure of $\p$.
      Base cases are constants; $A$ has two elements: one with path condition $\pc \land \app\sstate\bexp$; the other $\pc \land \neg\app\sstate\bexp$.
      Clearly, their disjunction is equivalent to $\pc$.

      For the inductive step $\p=\operator(\p_1,\dots,\p_n)$, following \Cref{def:symbolicsemantics} and letting $\P=\prod_{i\in[1..n]}\P_i$, we rewrite $\bigvee A$ as 
      $$ \bigvee_{\xi\in\P}[\pc\land\app\sstate{\guardof\operator}\land\bigwedge_{i\in[1..n]}\pc_\xi^{(i)}] \lor \bigvee_{\xi\in\P}[\pc\land\app\sstate{\neg\guardof\operator}\land\bigwedge_{i\in[1..n]}\pc_\xi^{(i)}]$$
      By distributive properties and since $\app\sstate{\guardof\operator}\lor\app\sstate{\neg\guardof\operator}\equiv\true$, this is equivalent to 
      \begin{align}\label{eq:coveragegoal}
        \pc \land \bigvee_{\xi\in\P} \bigwedge_{i=1}^n \pc_\xi^{(i)} 
      \end{align}
      We show that this in turn is equivalent to $\pc$.
      Clearly, \eqref{eq:coveragegoal} entails $\pc$.
      Conversely, let $\state\sat\pc$.
      Fix some $i\in[1..n]$ and consider the set 
      $ A_i := \{ \pc' \mid (\p_i,\sstate,\pc) \symbextr (\p',\sstate',\pc') \}$.
      By IH, $\pc \equiv \bigvee A_i$, so there is $(\p_i,\sstate,\pc)\symbextr(\p',\sstate',\pc')$ such that $\state\sat\pc'$.
      Put $(\p^{(i)},\sstate^{(i)},\pc^{(i)}):=(\p',\sstate',\pc')$.
      Doing this for every $i$ gives an $n$-tuple $\xi=(\p^{(i)},\sstate^{(i)},\pc^{(i)})_{i\in[1..n]}$ such that $\state \sat \bigwedge_{i=1}^n \pc_\xi^{(i)}$, and therefore 
      $\state\sat\bigvee_{\xi\in\P}\bigwedge_{i=1}^n\pc_\xi^{(i)}$.
  \end{itemize}
\end{proof}

\subsection{Proof of \Cref{thm:syncretebisim}} \label{app:syncretebisim}
\syncretebisim*
\begin{proof}
  For $\p\in\Prg$, $\sstate \in (\EE\Var)^\Var$, $\pc \in \BB\Var$, and $\state_0 \in \Val^\Var$, let $\Prop(\p,\sstate,\pc,\state_0)$ denote the property\footnote{Implication from right to left is really saying that if $(\p,\state_0\after{\sstate})\contr(\p',\state')$ then there is $(\p',\sstate',\pc')$ such that $(\p,\sstate,\pc)\symbextr(\p',\sstate',\pc')$, $\state_0\sat\pc'$, and $\state'=\state_0\after{\sstate'}$.} that 
  \begin{align}\label{eq:bisimproperty}
    \begin{array}{c}
    \exists(\p',\sstate',\pc') \textup{ s.t. } (\p,\sstate,\pc)\symbextr(\p',\sstate',\pc') \textup{ and } \state_0 \sat \pc' \\
      \iff \\
      (\p,\state_0\after{\sstate})\contr(\p',\state_0\after{\sstate'})
    \end{array}
  \end{align}
  Define the subset $\mathcal U \subseteq \Prg$ by $$ \{ \p \mid \forall (\sstate,\pc,\state_0)\in(\EE\Var)^\Var \times \BB\Var\times\Val^\Var \, (\state_0 \sat \pc \implies \Prop(\p,\sstate,\pc,\state_0) \} $$
  We show that $\Prg \subseteq \mathcal U$, by structural induction on elements of $\Prg$.
  Let $\operator \in \Sigma$ be an $n$-ary operator, and let $\p=\operator(\p_1,\dots,\p_n)$ with the induction hypotheses (IHs) that $\p_i \in \mathcal U$ for all $i\in[1..n]$.
  Let $(\sstate,\pc,\state_0) \in (\EE\Var)^\Var \times \BB\Var \times \Val^\Var$ be (arbitrary) such that $\state_0\sat\pc$.
  The proof of the biconditional $\Prop(\p,\sstate,\pc,\state_0)$ is twofold:
  \begin{itemize}
    \item[$\Rightarrow:$] Let $(\p,\sstate,\pc)\symbextr(\p',\sstate',\pc')$ with $\state_0\sat\pc'$. 
      By \Cref{def:symbolicsemantics}, there is $\xi=(\p^{(i)},\sstate^{(i)},\pc^{(i)})_{i\in[1..n]}$ such that $(\p_i,\sstate,\pc)\symbextr(\p^{(i)},\sstate^{(i)},\pc^{(i)})$.
      Let $W_\xi$ and the rules $\RR_{\xi,1}$ and $\RR_{\xi,2}$ be as in \Cref{def:symbolicsemantics}, with $\mt_j,\mc_j$ ($j=1,2$) their respective conclusion targets and outputs.
      Then either 
      $$ \textnormal{\rom 1} \qquad (\p',\sstate',\pc') = (\symmetasubst(\mt_1),\symmetasubst(\mc_1),\pc\land\app\sstate{\guardof\operator}\land\bigwedge_{i\in[1..n]}\pc^{(i)}) $$
      or
      $$ \textnormal{\rom 2} \qquad (\p',\sstate',\pc') = (\symmetasubst(\mt_2),\symmetasubst(\mc_2),\pc\land\app\sstate{\neg\guardof\operator}\land\bigwedge_{i\in[1..n]}\pc^{(i)} $$
      where $\symmetasubst=(\symmetasubst_X,\symmetasubst_S)$ and
      \begin{align}\label{eq:symmetasubst}
        \symmetasubst_X(\mx_i)=\p_i \qquad 
        \symmetasubst_X(\mx'_i)=p^{(i)} \qquad 
        \symmetasubst_S(\ma)=\sstate \qquad 
        \symmetasubst_S(\mb_i)=\sstate^{(i)}
      \end{align}
      In either case, \rom 1 or \rom 2, $\state_0\sat\pc^{(i)}$ for all $i$.
      
      By IHs, $\p_i\in\mathcal U$ for all $i\in[1..n]$ and since $(\p_i,\sstate,\pc)\symbextr(\p^{(i)},\sstate^{(i)},\pc^{(i)})$ and $\state_0 \sat \pc^{(i)}$, we have $(\p_i,\state_0\after{\sstate})\contr(\p^{(i)},\state_0\after\sstate^{(i)})$.
      Let $W = \{i\mid\p^{(i)}\in\Prg\}$ be the set indicating which subterms progress.
      Then $W=W_\xi$ and the rules $\RR_1,\RR_2$ from \Cref{def:concretesemantics} for this $W$ coincide with $\RR_{\xi,1}$ and $\RR_{\xi,2}$, and we obtain the same conclusion target and outputs $\mt_j,\mc_j$, $j=1,2$.
      Then, the meta-substitution $\metasubst=(\metasubst_X,\metasubst_S)$ of concrete states in \Cref{def:concretesemantics} is
      $$\metasubst_X(\mx_i)=\p_i \qquad 
      \metasubst_X(\mx'_i)=p^{(i)} \qquad 
      \metasubst_S(\ma)=\state_0\after\sstate \qquad 
      \metasubst_S(\mb_i)=\state_0\after\sstate^{(i)} 
      $$
      Therefore, $\p=\metasubst_X(\operator(\mx_1,\dots,\mx_n))=\symmetasubst_X(\operator(\mx_1,\dots,\mx_n))$, and $\metasubst_X(\mt_j) = \symmetasubst_X(\mt_j)$ for both $j=1,2$.
      Furthermore, since
      $$ \state_0 \after \symmetasubst_S(\mb) = \metasubst_S(\mb) $$
      for all $\mb \in \{\ma,\mb_1,\dots,\mb_n\}$, we moreover have
      $$\state_0 \after \symmetasubst_S(\mc_j) = \metasubst(\mc_j)$$ 
      for both $j=1,2$, by the Meta-Substitution \Cref{lem:metasubstlemma}.
      \begin{itemize}
        \item[-] Now if $(\p',\sstate',\pc')$ was as in \rom 1 then $\state_0\sat\app{\sstate}{\guardof\operator}$ which means, by the Substitution \Cref{lem:substlemma}, $\state_0\after\sstate\sat\guardof\operator$.
      Therefore, $\RR_1$ is used for the concrete $(\p,\state_0\after\sstate)\contr(\metasubst_X(\mt_1),\metasubst_S(\mc_1))$ and indeed $\metasubst_X(\mt_1)=\symmetasubst_X(\mt_1)=\p'$ and $\metasubst_S(\mc_1) = \state_0\after\symmetasubst_S(\mc_1)= \state_0\after\sstate'$, which was the goal of the conditional
        \item[-] Similar reasoning holds for when $(\p',\sstate',\pc')$ was as in \rom 2.
      \end{itemize}
    \item[$\Leftarrow:$] Let $(\p,\state_0\after\sstate)\contr(\p',\state')$ and $(\p_i,\state_0\after\sstate)\contr(\p^{(i)},\state^{(i)})$ for $i\in[1..n]$ (recall $\p=\operator(\p_1,\dots,\p_n)$).
      Without loss of generality, assume $\state_0\after\sstate\sat\guardof\operator$ so that the rule $\RR_1$ is used for this $\operator$ and $W$ (see \Cref{def:concretesemantics}), and let $\mt$ and $\mc$ be the conclusion target and output of $\RR_1$.
      Following \Cref{def:concretesemantics}, $\p' = \metasubst_X(\mt)$ and $\state'=\metasubst_S(\mc)$ where
      $$\metasubst_X(\mx_i)=\p_i \qquad 
      \metasubst_X(\mx'_i)=p^{(i)} \qquad 
      \metasubst_S(\ma)=\state_0\after\sstate \qquad 
      \metasubst_S(\mb_i)=\state^{(i)} 
      $$
      Since (IHs) every $\p_i\in\mathcal U$, we have, for every $i$, $(\p_i,\sstate,\pc) \symbextr (\p^{(i)},\sstate^{(i)},\pc^{(i)})$ for some $(\sstate^{(i)},\pc^{(i)})_{i\in[1..n]}$ such that $\state^{(i)} = \state_0\after\sstate^{(i)}$ and $\state_0\sat\pc^{(i)}$.
      But this means that $W_\xi=W$ for $\xi=(\p^{(i)},\sstate^{(i)},\pc^{(i)})_{i\in[1..n]}$, yielding the same rule $\RR_1$ and justifying the transition $(\p,\sstate,\pc)\symbextr(\symmetasubst_X(\mt),\symmetasubst_S(\mc),\pc')$ with $\pc'=\pc\land\app\sstate{\guardof\operator}\land\bigwedge_{i\in[1..n]}\pc^{(i)}$, where $\symmetasubst_X$ and $\symmetasubst_S$ are exactly as in \eqref{eq:symmetasubst}.
      \begin{itemize}
        \item $\symmetasubst_X(\mt)=\metasubst_X(\mt)=\p'$;
        \item $\state_0\after\symmetasubst_S(\mc)=\metasubst_S(\mc)=\state'$ by the same reasoning as in the proof of the other direction (using the Meta-Substitution \Cref{lem:metasubstlemma}); and 
        \item $\state_0\sat\pc'$, because \rom 1 $\state_0\sat\pc$ by assumption, \rom 2 $\state_0\after\sstate\sat\guardof\operator$ (assumed earlier w.l.o.g.) and so $\state_0\sat\app\sstate{\guardof\operator}$ by \Cref{lem:substlemma}, and \rom 3 $\state_0\sat\pc^{(i)}$ for all $i$.
      \end{itemize}
      Therefore we found a transition $(\p,\sstate,\pc) \symbextr (\p',\sstate',\pc')$ such that $\state'=\state_0\after\sstate'$ and $\state_0\sat\pc'$, and we are done.
  \end{itemize}
  Thus, since $\sstate$, $\pc$, and $\state_0$ were arbitrary, we conclude $\p\in\mathcal U$ and induction is finished.
  Hence, $\Prg \subseteq \mathcal U$.
\end{proof}
\end{document}